\documentclass[a4paper,10pt]{article}
\usepackage[utf8]{inputenc}

\usepackage{amsmath,amssymb,amsthm}
\usepackage{paralist}
\usepackage{tikz}

\newcommand{\A}{\mathcal{A}}
\newcommand{\B}{\mathcal{B}}
\newcommand{\C}{\mathcal{C}}
\newcommand{\Z}{\mathbb{Z}}
\newcommand{\splt}{\operatorname{\sf Split}}
\newcommand{\bsc}{\operatorname{\sf Basic}}
\newcommand{\id}{\operatorname{Id}}
\providecommand{\Omicron}{\mathrm{O}}

\newtheorem{theorem}{Theorem}[section]
\newtheorem{lemma}[theorem]{Lemma}
\newtheorem{corollary}[theorem]{Corollary}
\newtheorem{proposition}[theorem]{Proposition}

\title{Subset synchronization of DFAs and PFAs, and some other results}
\author{Michiel de Bondt}

\begin{document}

\maketitle

\begin{abstract}
This paper contains results which arose from the research which led to 
arXiv:1801.10436, but which did not fit in arXiv:1801.10436. So 
arXiv:\allowbreak1801.10436 contains the highlight results, but there 
are more results which are interesting enough to be shared.
\end{abstract}

\section{Introduction}

For an NFA (nondeterministic finite automaton) $\A = (Q,\Sigma,\cdot)$,
$Q$ and $\Sigma$ are sets, and $\cdot: Q \times \Sigma \rightarrow 2^Q$.
A complete NFA is an NFA $\A = (Q,\Sigma,\cdot)$ for which 
$\cdot: Q \times \Sigma \rightarrow 2^Q \setminus \varnothing$. Here, 
the members of $Q$ are called states, and $\Sigma$ is called an alphabet 
of letters or symbols. $\cdot$ is called the transition function in some
sources, and sometimes the symbol $\delta$ is used for it.

A PFA (partial finite automaton) is in fact
an NFA for which $\cdot$ maps to sets of size at most $1$, but we will
use a different syntax, namely $\cdot: Q \times \Sigma \rightarrow 
Q \cup \{\bot\}$. A DFA (deterministic finite automaton) is a PFA which 
is in fact also a complete NFA, so we have 
$\cdot: Q \times \Sigma \rightarrow Q$ in that case.

$\cdot$ is left-associative, and we will omit it mostly. We additionally 
define $\cdot: 2^Q \times \Sigma \rightarrow 2^Q$ for NFAs 
$\A = (Q,\Sigma,\cdot)$, by $S a = \bigcup_{s \in S} sa$. For PFAs
(and DFAs), we do this as well, but in a totally different way, namely
by $S a = \bigcup_{s \in S} \{sa\}$ if $sa \ne \bot$ for all $s \in S$,
and $S a = \varnothing$ otherwise. Notice that both definitions agree 
with each other for DFAs, but not for proper PFAs.

We define $qw$ inductive as follows for states $q \in Q$, subsets 
$S \subseteq Q$, and words $w \in \Sigma^{*}$:
\begin{align*}
q \lambda &= q & q (x w) &= (qx)w & S \lambda &= S & S (x w) &= (Sx)w
\end{align*}
Here, $\lambda$ is the empty word, $x$ is the first letter of the word 
$xw$ and $w$ is the rest of $xw$.

We say that a complete NFA $\A = (Q,\Sigma,\cdot)$ is synchronizing 
(in $l$ steps), if there exists a $w \in \Sigma^{*}$ (of length $l$), 
such that $Qw$ has size $1$.
We say that a PFA $\B = (Q,\Sigma,\cdot)$ is carefully synchronizing
(in $l$ steps), if there exists a $w \in \Sigma^{*}$ (of length $l$), 
such that $Qw$ has size $1$. 
We say that a DFA is synchronizing (in $l$ steps) if it is carefully 
synchronizing as a PFA (in $l$ steps), or equivalently, if it is 
synchronizing as a complete NFA (in $l$ steps).

In \cite{arXiv:1801.10436}, there are many results about DFAs and PFAs 
in connection with synchronization, each of which is considered a 
relatively important result by at least one of the authors. But 
there are more results, and although they may be less important, they
are still important enough to be shared. For many of them, we need to 
extend the definition of (careful) synchronization to subsets of the 
state set. 

We say that a complete NFA $\A = (Q,\Sigma,\cdot)$ is synchronizing 
on $S \subseteq Q$ (in $l$ steps), if there exists a $w \in \Sigma^{*}$ 
(of length $l$), such that $Sw$ has size $1$.
We say that a PFA $\B = (Q,\Sigma,\cdot)$ is carefully synchronizing on 
$S \subseteq Q$ (in $l$ steps), if there exists a $w \in \Sigma^{*}$ 
(of length $l$), such that $Sw$ has size $1$. 
We say that a DFA is synchronizing on $S \subseteq Q$ (in $l$ steps) 
if it is carefully synchronizing on $S$ as a PFA (in $l$ steps),
or equivalently, if it is synchronizing on $S$ as a complete NFA 
(in $l$ steps).

In section \ref{subsetsync5states}, we give all possible maximum 
subset synchronization lengths for all PFAs and all DFAs with up to
$5$ states, and all subsets of these states. These results arise 
from a combination of computations and reasoning. The actual computations
can be found along with the code of \cite{arXiv:1801.10436}. The 
computations by itself are already sufficient to find all possible 
maximum subset synchronization lengths for all \emph{synchronizing}
DFAs with up to $6$ states and to prove Cardoso's conjecture for $6$ 
states. Cardoso's conjecture is Conjecture 7 of \cite{arXiv:1607.04025},
and was already proved for up to $5$ states in \cite{arXiv:1607.04025}.

In section \ref{subsetsync}, we give the maximum synchronization
lengths for state subsets of complete NFAs. More precisely, we show 
that we can force that all subsets which are not proper supersets of 
the start subset need to be traversed to synchronize a specific (start)
subset of size at least $2$. This extends a result of \cite{Burkhard} 
to state subsets. The alphabet size of our construction is less than 
$\frac12 n^2$. Furthermore, we give 
the maximum synchronization lengths for state subsets of size at most 
$3$ of PFAs. More precisely, we show that we can force that all subsets 
of size $3$ and $2$ need to be traversed to synchronize a specific
subset of size $3$. The Cerny automaton has the property that all 
subsets of size $2$ need to be traversed to synchronize a specific
subset of size $2$.

In section \ref{subsetsyncasymp}, we give asymptotic lower bounds 
for the maximum subset synchronization lengths of PFAs and 
DFAs (where the size of the subset is not arbitrary, 
but chosen to obtain the best lower bound). 
With this, we also take transitivity of the automaton into account. 
A PFA $\B = (Q,\Sigma,\cdot)$ (in particular a DFA) is transitive, 
if for every pair $(q,q') \in Q^2$, there exists a $w \in \Sigma^{*}$ 
such that $qw = q'$. We start with considering PFAs and DFAs with
an arbitrary number of symbols. The results about them are formulated 
as propositions instead of theorems, because they follow more or less
directly from techniques by others. Next, we consider binary PFAs
and binary DFAs. The result for binary PFAs is obtained by adapting the 
construction of \cite{arXiv:1801.10436}: the number of states is 
reduced and there is a finishing symbol instead of a start symbol.
The result for binary DFAs is obtained by combining the techniques of the 
result for binary PFAs with some of those in \cite{arXiv:1403.3972}. 
\cite{arXiv:1403.3972} contains lower bounds for the maximum subset 
synchronization lengths of (transitive) binary DFAs as well, but the 
lower bounds in this paper are better.

Section \ref{d3nfa2pfa} is about the property of D3-directing for NFAs. 
We discuss Lemma 3 of \cite{d3nfa}, which states that the maximum
length of a D3-directing word of an NFA with $n$ states is the same as 
that of a PFA with $n$ states. Furthermore, we discuss Theorem 1 of 
\cite{arXiv:1703.07995}, which states that the maximum length of a 
D3-directing word of a complete NFA with $n$ states is the same as that
of a DFA with $n$ states. The proof of Lemma 3 of \cite{d3nfa} is 
actually too short, and the proof of Theorem 1 of \cite{arXiv:1703.07995} 
is actually too long. We give a combined proof of which the length is 
in between. Furthermore, we compute the 
minimum alphabet size for D3-directing NFAs which take the maximum number
of steps, up to $7$ states. Again, the actual computations can be found 
along with the code of \cite{arXiv:1801.10436}. Section \ref{d3nfa2pfa} does not 
discuss the subset variant of D3-directing and (careful) synchronization, 
but the results in it can be generalized to subset variants in a 
straightforward manner.

In section \ref{martyuginprime}, we show that the prime number 
construction for careful synchronization of PFAs in \cite{primepfa} by 
Martyugin yields a synchronization length which is strictly between 
polynomial and exponential in the number of states. Proving this was
not the only reason to redo the complexity estimation in \cite{primepfa}. 
The other reason is that the estimation is not so accurate on the first 
point, and overly accurate on other points, to be not so accurate over all. 
Our estimate is more accurate on the first point, and not more accurate 
than needed on other points, to be more accurate over all.

\section{Subset synchronization up to 5 states}
\label{subsetsync5states}

In our algorithm for computing slowly synchronizing automata, the synchronization 
estimate for a subset $S$ is determined roughly as follows. 
First, the length of a path from $S$ to a smaller subset $S'$ is estimated. 
Here $S'$ does not need to be determined; even the size of $S'$ may be undetermined. 
Next, the length of a synchronizing path from $S'$ is estimated recursively.

Here, it is assumed that $S'$ does indeed synchronize. But with subset synchronization,
$S'$ does not need to synchronize. To overcome this problem, we just assume that subsets
of size less than that of $S$, which we denote by $|S|$, synchronize.

The DFA algorithm for computing slowly synchronizing automata is constructed in such a way, 
that it only gives solutions for which all pairs synchronize, even if we start with a 
strict subset of the states. So all given solutions are synchronizing automata, and $S'$ 
above will synchronize all the time.

The PFA algorithm for computing slowly synchronizing automata on subset $S$ will give all 
solutions in which subsets of size less than $|S|$ synchronize, but it may possibly give 
some other solutions as well.

\subsection{DFAs up to 5 states}

We found the following subset synchronizations lenghts for fully synchronizing DFAs:
\begin{center}
\renewcommand{\arraystretch}{1.4}
\begin{tabular}{|c|cccc|}
\hline
$|S|$ & $n = 3$ & $n = 4$ & $n = 5$ & $n = 6$ \\
\hline
2 & 3 & 6 & 10 & 15 \\[-5pt]
3 & 4 & 8 & 13 & 20 \\[-5pt]
4 & & 9 & 15 & 22 \\[-5pt]
5 & & & 16 & 24 \\[-5pt]
6 & & & & 25 \\
\hline
\end{tabular} 
\end{center}
If you see a pattern in these values, then take special attention to the value
$20$ for $n = 6$ and $|S| = 3$. Actually, there is a pattern, namely the known 
formula
$$
(n-1)^2 \Big(\Big\lceil \frac{n}{|S|} \Big\rceil - 1 \Big) 
\Big( 2n - |S| \Big\lceil \frac{n}{|S|} \Big\rceil - 1 \Big) 
$$
for the synchronization lengths of the Cerny automata. A conjecture of 
\^{A}ngela Cardoso asserts that the subset synchronization lengths of the Cerny
automata are the best possible for synchronizing DFAs, see \cite{arXiv:1607.04025}.
So we have proved this conjecture up to $6$ states. The conjecture is trivially
true for $|S| = 2$, because in order to synchronize the hardest pair in a 
Cerny automaton, all pairs must be traversed.

\begin{theorem}
Suppose that $n \le 5$ and that subset $S$ of the $n$ states synchronizes. 
If our DFA is not synchronizing, then the length of the minimum synchronizing
word for $S$ is less than the value for $n$ and $|S|$ in the above table.

In particular, the above table yields the subset synchronization lengths for 
DFAs with $n \le 5$ states.
\end{theorem}

\begin{proof}
If $|S| = 2$, then the values in the table indicate that the corresponding
synchronization paths visit all subsets of size $2$. In particular, 
our DFA is synchronizing. 

If $|S| = n$, then our DFA is synchronizing as well. So three cases remain.
\begin{itemize}

\item \emph{$n = 4$ and $|S| = 3$.}

Assume that our DFA is not synchronizing. Then there is a subset of size $2$
which does not synchronize, say that $\{3,4\}$ does not synchronize.
Then there are at most $\binom{4}{2} - 1 = 5$ subsets of size $2$ which do 
synchronize. 

Since $\{1,3,4\}$ and $\{2,3,4\}$ do not synchronize
either, there are at most $\binom{4}{3} - 2 = 2$ subsets of size $3$
which do synchronize. This leaves $2 + 5 = 7$ subsets for the 
synchronization path, which is less than the value $8$ for $n$ and $|S|$
in the above table.

\item \emph{$n = 5$ and $|S| = 4$.}

Assume that our DFA is not synchronizing. Then there is a subset of size $2$
which does not synchronize, say that $\{4,5\}$ does not synchronize.
We distinguish three subcases.

\begin{itemize}

\item $\{4,5\}$ is the only subset of size $2$ which does not synchronize.

Then every symbol of our DFA acts as a permutation on $\{4,5\}$.
So the number of states of the set $\{4,5\}$ in our synchronization path
will be increasing. We can count the number of subsets $S'$ for given
values of $|S'|$ and $|S' \cap \{4,5\}|$
\begin{center}
\renewcommand{\arraystretch}{1.4}
\begin{tabular}{|c|ccc|}
\hline
$|S' \cap \{4,5\}|$ & $|S'] = 4$ & $|S'| = 3$ & $|S'| = 2$ \\
\hline
0 & 0 & 1 & 3 \\[-5pt]
1 & 2 & 6 & 6 \\
\hline
\end{tabular} 
\end{center}
From this, we deduce that the synchronization path has length at most
$2 + 6 + 6 = 14$, which is less than the value $15$ for $n = 5$ and $|S| = 4$
in the table.

\item $\{3,4\}$ and $\{4,5\}$ do not synchronize.

Then one can count that at most one subset of size $4$ synchronizes,
at most $5$ subsets of size $3$ synchronize, and at most $8$ subsets
of size $2$ synchronize. From this, we deduce that the synchronization 
path has length at most $1 + 5 + 8 = 14$, which is less than the value 
$15$ for $n = 5$ and $|S| = 4$ in the table.

\item $\{2,3\}$ and $\{4,5\}$ do not synchronize.

Then there is no subset of size $4$ which synchronizes. This 
contradicts that $|S|$ synchronizes.

\end{itemize}

\item \emph{$N = 5$ and $|S| = 3$.}

Just as in the previous case, we assume that $\{4,5\}$ does not 
synchronize, and we distinguish the same three subcases.
\begin{itemize}

\item $\{4,5\}$ is the only subset of size $2$ which does not synchronize.

In a similar manner as in the corresponding subcase of the case above,
we deduce that the synchronization path has length at most
$1 + 6 + 6 = 13$, which is exactly the value $13$ for $n = 5$ and $|S| = 3$
in the table.

Supppose that length $13$ is indeed possible. We derive a contradiction.
We can deduce that $S = \{1,2,3\}$. Let $a$ be the first symbol of 
the shorthest synchronizing word for $S$, and let $S' = Sa$. 
Then $|S' \cap \{4,5\}| = 1$, say that $S' \cap \{4,5\} = \{4\}$. 

The preimage $\{4\}a^{-1}$ of $\{4\}$ under $a$ contains exactly one element of 
$\{1,2,3\}$ and exactly one element of $\{4,5\}$. There are $2$ subsets
of size $3$, which contain $\{4\}a^{-1}$, but not $\{4,5\}$. Both subsets
are mapped by $a$ to a subset of size $2$ which synchronizes. This contradicts
that there are $7$ subsets of size $3$ in our synchronization path.

\item $\{3,4\}$ and $\{4,5\}$ do not synchronize.

In a similar manner as in the corresponding subcase of the case above,
we deduce that the synchronization path has length at most
$5 + 8 = 13$, which is exactly the value $13$ for $n = 5$ and $|S| = 3$
in the table.

Supppose that length $13$ is indeed possible. We derive a contradiction.
Notice that $\{1,2\}$ is in the synchronization path, say $Sw = \{1,2\}$.
Since either $S \cap \{3,4\} \ne \varnothing$ or $S \cap \{4,5\} \ne \varnothing$, 
we deduce that either $\{3,4\}w \cap \{1,2\} \ne \varnothing$ or 
$\{4,5\}w \cap \{1,2\} \ne \varnothing$. So there is another subset of size 
$2$ which does not synchronize. Contradiction.

\item $\{2,3\}$ and $\{4,5\}$ do not synchronize.

Then one can count that at most $4$ subsets of size $3$ synchronize,
and at most $8$ subsets of size $2$ synchronize. From this, we deduce that 
the synchronization path has length at most $4 + 8 = 12$, which is less than 
the value $13$ for $n = 5$ and $|S| = 3$ in the table. \qedhere

\end{itemize}
\end{itemize}
\end{proof}

\subsection*{PFAs up to 5 states}

We found the following subset synchronizations lengths for PFAs, for which
subsets of size less than $|S|$ synchronize, and possibly other PFAs.
\begin{center}
\renewcommand{\arraystretch}{1.4}
\begin{tabular}{|c|ccc|}
\hline
$|S|$ & $n = 3$ & $n = 4$ & $n = 5$ \\
\hline
2 & 3 & 6 & 10 \\[-5pt]
3 & 4 & 10 & 20 \\[-5pt]
4 & & 10 & 22 \\[-5pt]
5 & & & 21 \\
\hline
\end{tabular} 
\end{center}
Notice that the value $22$ for $n = 5$ and $|S| = 4$ is bigger than the value
$21$ for $n = 5 = |S|$. Hence the automaton for the case $n = 5$ and $|S| = 4$
is not synchronizing.

Just as for the DFAs, we can take the Cerny automata if $|S| = 2$. If $|S| = n$,
then we take the PFAs in \cite{arXiv:1801.10436}. Otherwise, we take the following
ternary PFAs:
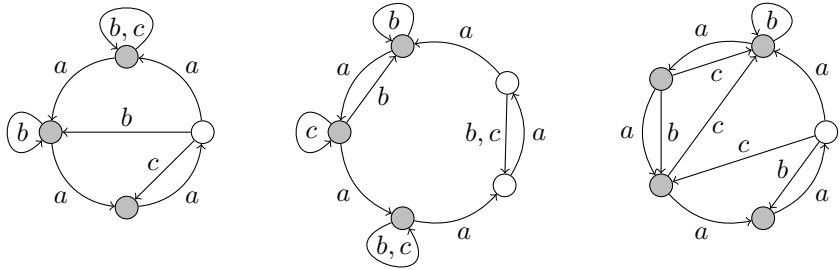
\begin{center}
\begin{tikzpicture}
\begin{scope}[shift={(0,0)}]
\node[circle,draw,fill=lightgray,inner sep=0pt,minimum width=3mm] (0) at (90:1) {};
\node[circle,draw,fill=lightgray,inner sep=0pt,minimum width=3mm] (1) at (180:1) {};
\node[circle,draw,fill=lightgray,inner sep=0pt,minimum width=3mm] (2) at (270:1) {};
\node[circle,draw,inner sep=0pt,minimum width=3mm] (3) at (360:1) {};
\foreach \n [remember=\n as \nlast (initially 3)] \n in {0,1,2,3} {
\draw[->,bend right=39] (\nlast) edge node[shift={(90*\n+45:2mm)}] {$a$} (\n);
}
\draw[->] (0) edge[out=40,in=140,looseness=12] node[below] {$b,c$} (0);
\draw[->] (1) edge[out=135,in=225,looseness=10] node[right] {$b$} (1);
\draw[->] (3) edge node[pos=0.50,above,inner sep=2pt] {$b$} (1);
\draw[->] (3) edge node[pos=0.70,above,inner sep=4pt] {$c$} (2);
\end{scope}
\begin{scope}[shift={(4,0)}]
\node[circle,draw,fill=lightgray,inner sep=0pt,minimum width=3mm] (0) at (108:1.2) {};
\node[circle,draw,fill=lightgray,inner sep=0pt,minimum width=3mm] (1) at (180:1.2) {};
\node[circle,draw,fill=lightgray,inner sep=0pt,minimum width=3mm] (2) at (252:1.2) {};
\node[circle,draw,inner sep=0pt,minimum width=3mm] (3) at (324:1.2) {};
\node[circle,draw,inner sep=0pt,minimum width=3mm] (4) at (393:1.2) {};
\foreach \n [remember=\n as \nlast (initially 4)] \n in {0,1,2,3,4} {
\draw[->,bend right=30] (\nlast) edge node[shift={(72*\n+72:2mm)}] {$a$} (\n);
}
\draw[->] (0) edge[out=63,in=153,looseness=10] node[pos=0.45,below,inner sep=2.5pt] {$b$} (0);
\draw[->] (1) edge[out=135,in=225,looseness=10] node[right] {$c$} (1);
\draw[->] (2) edge[out=202,in=302,looseness=12,pos=0.55,inner sep=2.5pt] node[above] {$b,c$} (2);
\draw[->] (1) edge node[pos=0.40,right,inner sep=4pt] {$b$} (0);
\draw[->] (4) edge node[pos=0.50,left,inner sep=2pt] {$b,c$} (3);
\end{scope}
\begin{scope}[shift={(8,0)}]
\node[circle,draw,fill=lightgray,inner sep=0pt,minimum width=3mm] (0) at (72:1.2) {};
\node[circle,draw,fill=lightgray,inner sep=0pt,minimum width=3mm] (1) at (144:1.2) {};
\node[circle,draw,fill=lightgray,inner sep=0pt,minimum width=3mm] (2) at (216:1.2) {};
\node[circle,draw,fill=lightgray,inner sep=0pt,minimum width=3mm] (3) at (288:1.2) {};
\node[circle,draw,inner sep=0pt,minimum width=3mm] (4) at (360:1.2) {};
\foreach \n [remember=\n as \nlast (initially 4)] \n in {0,1,2,3,4} {
\draw[->,bend right=30] (\nlast) edge node[shift={(72*\n+36:2mm)}] {$a$} (\n);
}
\draw[->] (0) edge[out=27,in=117,looseness=10] node[pos=0.55,below,inner sep=2.5pt] {$b$} (0);
\draw[->] (1) edge node[pos=0.55,below,inner sep=3pt] {$c$} (0);
\draw[->] (2) edge node[pos=0.4,right,inner sep=3pt] {$c$} (0);
\draw[->] (1) edge node[pos=0.5,right,inner sep=2pt] {$b$} (2);
\draw[->] (4) edge node[pos=0.5,above,inner sep=2.5pt] {$c$} (2);
\draw[->] (4) edge node[pos=0.4,left,inner sep=4pt] {$b$} (3);
\end{scope}
\end{tikzpicture}
\end{center}

\begin{theorem}
Suppose that $n \le 5$ and that subset $S$ of the $n$ states synchronizes. 
If our PFA does not synchronize all subsets of size less than $|S|$, then the
length of the minimum synchronizing word for $S$ is less than the values 
for $n$ and $|S|$ in the above table.

In particular, the above table yields the subset synchronization lengths for 
PFAs with $n \le 5$ states.
\end{theorem}

\begin{proof}
If $|S| \le 3$, then the values in the table indicate that the corresponding
synchronization paths visit all subsets of size at least $2$ and at most $|S|$. 
In particular, all subsets of size less than $|S|$ synchronize.

If $|S| = n$, then our PFA synchronizes, and so do all subsets of size less 
than $|S|$.

So the only case which remains is $n = 5$ and $|S| = 4$. 

So assume we have a PFA with $n = 5$ states, which synchronizes a subset 
$S$ of size $4$. Now suppose that there is a subset $S'$ of size less than $4$, 
which does not synchronize. Then we can choose such an $S'$ of size $3$. Say
that $\{3,4,5\}$ does not synchronize. Then $\{1,3,4,5\}$ and $\{2,3,4,5\}$
do not synchronize either. So there are at most
$$
\big(\tbinom{5}{4} - 2\big) + \big(\tbinom{5}{3} - 1\big) + \tbinom{5}{2}
= 3 + 9 + 10 = 22
$$
subsets of size $2$, $3$ or $4$ which do synchronize. 

It follows that a synchronization path of minimum length at least $22$ for $S$
contains all subsets of size $2$ and $3$, except $S'$. This is
however not possible. Indeed, let $a$ be the symbol in the synchronization 
path which is applied on the last subset of size $4$, say $\{q_1,q_2,q_3,q_4\}$.
Then $\{q_1,q_2,q_3,q_4\}a$ has size less than $4$. Say that $q_1 a = q_2 a$.

Then $\{q_1,q_2,q_3\}a$ and $\{q_1,q_2,q_4\}a$ have size less than $3$.
Since our synchronization path contains all subsets of size $2$, we deduce that
$\{q_1,q_2,q_3\} \neq S' \neq \{q_1,q_2,q_4\}$. As our synchronization path does 
not require both $\{q_1,q_2,q_3\}$ and $\{q_1,q_2,q_4\}$, there is another 
subset of size $3$ besides $S'$ which can be excluded from the synchronization
path. Contradiction, so our PFA synchronizes all subsets of size less than $|S|$.
\end{proof}

\section{Subset synchronization with any number of states}
\label{subsetsync}

\subsection{Complete NFAs}

We construct a complete NFA with $\Omicron(n^2)$ symbols in which every subset 
of size at least $2$ must be traversed before a singleton can be reached. 
But first, we construct a complete NFA with only $\Omicron(n)$ symbols in which 
half of the subsets is traversed in a synchronization path.

\begin{theorem}
Let $n \ge 2$. There exists a complete NFA with state set $Q = \{1,2,\ldots,n\}$ 
and $2n - 2$ symbols, which synchronizes $Q$ in $2^{n-1} - 1$ steps and 
$\{n-1,n\}$ in $2^{n-2}$ steps.

More precisely, the shortest path from $Q$ to $\{1\}$ traverses all subsets of
$Q$ except $\varnothing$ in reverse lexicographic order. Furthermore, the NFA is
transitive on the set of nonempty state subsets.
\end{theorem}

\begin{proof}
To get from a subset without $1$ to its successor, we define
$$
(n a_i, (n-1) a_i, \ldots, 1 a_i) = (\{n\}, \{n-1\}, \ldots, \{i+1\}, 
\{1,2,\ldots,i-1\}, Q, \ldots, Q)
$$
for each $i \in \{2,3,\ldots,n\}$. The successor of a subset with $1$ is obtained 
by removing $1$, for which we use symbols defined by
$$
(n b_i, (n-1) b_i, \ldots, 1 b_i) = (\{n\}, \{n-1\}, \ldots, \{2\}, \{i\}) 
$$
for each $i \in \{2,3,\ldots,n\}$. The last claim follows from $\{1\}a_2 = Q$
\end{proof}

\begin{theorem}
Let $n \ge 3$. There exists a complete NFA with state set $Q = \{1,2,\ldots,n\}$ 
and $\frac12 n^2 - \frac12 n + 1$ symbols, in which the shortest path from $Q$ to $\{3\}$ 
first traverses all subsets of size at least $2$ of $Q$ in reverse lexicographic order,
and next traverses $\{2\}, \{1\}, \{n\}, \{n-1\}, \ldots, \{3\}$, in that order.
Furthermore, the NFA is transitive on the set of nonempty state subsets.
\end{theorem}

\begin{proof}
Just as before, to get from a subset of size at least $2$ without $1$ to its 
successor, we define
$$
(n a_i, (n-1) a_i, \ldots, 1 a_i) = (\{n\}, \{n-1\}, \ldots, \{i+1\}, 
\{1,2,\ldots,i-1\}, Q, \ldots, Q)
$$
but only for each $i \in \{2,3,\ldots,n-1\}$. For the successor of a subset of
size at least $2$ with $1$, we use symbols $b_{22}$ and $b_{ij}$ with 
$3 \le i < j \le n$, defined by
$$
(n b_{ij}, (n-1) b_{ij}, \ldots, 1 b_{ij}) = (\{n\}, \{n-1\}, \ldots, \{2\}, \{i,j\}) 
$$
and symbols defined by $i c_i = \{i-1\}$ and
$$
(n c_i, (n-1) c_i, \ldots, 1 c_i) = (Q,\ldots,Q,\{i-1\},Q,\ldots,Q,
\{1,2,\ldots,i-2\})
$$
for each $i \in \{3,4,\ldots,n\}$. If $i \ge 4$, then symbol $c_i$ sends $\{i\}$ 
to its successor, too. Since $a_2$ sends $\{2\}$ to its successor, we only need a 
symbol which sends $\{1\}$ to its successor. For that purpose, we define 
$c_1$ by $i c_1 = \{n\}$ if $i = 1$ and $i c_1 = Q$ otherwise. The last claim
follows from $\{3\}c_1 = Q$.
\end{proof}

\begin{corollary}
The maximum synchronization length for state subsets of size $|S| \ge 2$ 
in complete NFAs with $n$ states is 
$$
2^n - n - 2^{n-|S|}
$$
\end{corollary}

\begin{proof}
Let $|S|$ be a subset of the $n$ states of an NFA. In a synchronization path
of $|S|$, the $2^{n-|S|} - 1$ proper supersets of $S$ do not occur. $\varnothing$
does not occur either, and singleton sets only occur at the end. This
reduces the first claim to the last claim. The last claim follows by taking 
$S = \{n-|S|+1,n-|S|+2,\ldots,n\}$ in the above theorem.
\end{proof}

In \cite{arXiv:1801.05391}, careful synchronization is defined for NFAs in general, not 
just PFAs. For that type of synchronization, we have the same results as in the 
above corollary. This is because we can turn an NFA which synchronizes some subset 
$S$ carefully into a complete NFA as follows: we replace each $ix = \varnothing$ 
by $ix = Q$, where $i$ is any state, $x$ is any symbol, and $Q$ is the subset of 
all states.

\subsection{PFAs and subsets of size 3}

The following theorem yields the lengths of $3$-set synchronization for PFAs.

\begin{theorem}
Let $n \ge 3$. Then there exists a PFA with $n$ states and $2n -3$ symbols,
say with states $1,2,\ldots,n$, such that the shortest path from
$\{n-2,\allowbreak n-1,n\}$ to $\{n\}$ traverses all subsets of size $1$, $2$ and $3$.
\end{theorem}

\begin{proof}
We traverse the subsets of size $1$ and $2$ in lexicographic order. For that
purpose, we define
$$
(1 a_i, 2 a_i, \ldots, n a_i) = 
\begin{cases}
(1,\ldots,i-1,i+1,\bot,\ldots,\bot,i+2) & 1 \le i \le n-2 \\
(1,\ldots,i-1,i+1,i) & i = n-1
\end{cases}
$$
Traversing the subsets of size $3$ in the same order is impossible, since it would 
yield $n-3$ shortcuts for the subsets of size $2$: subset $\{i,n\}$ can be skipped for 
all $i \le n-3$. The order of the subsets of size $3$ will be lexicographic as well, but
first, we negate the first state of the $3$-set. So subset $\{n-2,n-1,n\}$ is ordered
by way of $\{2-n,n-1,n\}$ and is the first $3$-set, and subset $\{1,n-1,n\}$ is ordered
by way of $\{-1,n-1,n\}$ and is the last $3$-set. 

With the above symbols, we can already move subsets of size $3$ to their successors, 
provided the first state does not need to be changed for that. 
To move other $3$-sets to their successors and to move the last subset of size
$2$ and $3$ to the first subset of size $1$ and $2$ respectively, we add the following 
symbols.
$$
(0 b_i, 1 b_i, \ldots, n b_i) =
\begin{cases}
(i+1, \bot, \ldots, \bot, i, i) & i = 1 \\
(1,\ldots,i-2,\bot,i+1,\bot,\ldots,\bot,i,i-1) & 2 \le i \le n-2
\end{cases}
$$
It is a straightforward exercise to verify that no shortcuts are possible.
\end{proof}

\begin{corollary}
The maximum synchronization length for state subsets of size $3$ in PFAs with 
$n$ states is 
$$
\binom{n}{3} + \binom{n}{2}
$$
This maximum cannot be obtained for PFAs which synchronize a state subset
of size $4$.
\end{corollary}

\begin{proof}
The first claim follows from the above theorem. To prove the last claim,
let $a$ be a symbol which reduces a subset of size $4$, say $\{q_1,q_2,q_3,q_4\}$,
to a smaller subset. Say that $q_1 a = q_2 a$. Then both $\{q_1,q_2,q_3\}a$ and
$\{q_1,q_2,q_4\}a$ have size less than $3$. Hence not all subsets of size $3$
need to be traversed.
\end{proof}

Notice that the above corollary proves that the maximum synchronization length
$p(4)$ of a PFA with $4$ states is at most $10$, and that $p(4) = 10$.

\section{Asymptotics for subset synchronization of automata}
\label{subsetsyncasymp}

We start with some results for an arbitrary number of symbols.

\begin{proposition} \label{csubscl}
There exists a transitive PFA $M$ with $n$ states,
which has a subset which synchronizes carefully in 
$$
\Omega \big(2^n/\sqrt{n}\big) = \Omega (1.9999^n)
$$
steps.
\end{proposition}

\begin{proof}
Let $S_1, S_2, \ldots, S_b$ be the subsets of size $\lfloor n/2 \rfloor$
of the $n$ states. Take for each $i < b$ a symbol $a_i$ such that
$S_i a_i = S_{i+1}$ and $a_i$ is undefined outside $S_i$. Take 
$a_b$ such that $S_b a_b$ is a singleton and $a_b$ is undefined outside $S_b$.
The word $a_1 a_2 \ldots a_b$ is carefully synchronizing for subset $S_1$, and
is a prefix of every other such word. For even $n$, the estimate 
$\Theta(2^n / \sqrt{n})$ for $b$ can be found on the internet. 
For odd $n$, $b$ is half of the value it would have if $n$ was $1$ larger.

Take $a_0$ such that $S_b a_0 = S_1$ and $a_0$ is undefined outside $S_b$.
Then the word $a_1 a_2 \ldots a_b$ is maintained as the shortest carefully 
synchronizing word. The orbit of a state $q$ has size
$> \lceil n/2 \rceil$, because it intersects with every subset of size 
$\lfloor n/2 \rfloor$. So this orbit contains a subset of 
$\lfloor n/2 \rfloor$ states, hence it contains all states. 
So $M$ is transitive.
\end{proof}

The first claim of the following result can be found in \cite{Burkhard}.

\begin{proposition}
There exists a DFA $M$ with $n$ states, 
which has a subset which synchronizes in 
$$
\Omega \big(2^n/\sqrt{n}\big) = \Omega (1.9999^n)
$$
steps. Furthermore, the DFA is transitive up to $2$ sink states.

There exists a transitive DFA $M$ with $n$ states, 
which has a subset which synchronizes in 
$$
\Omega \big(2^{n/2}/\sqrt{n}\big) = \Omega (1.4142^n)
$$
steps.
\end{proposition}

\begin{proof}
The first claim follows from proposition \ref{csubscl} above and 
Lemma 1 of \cite{arXiv:1403.3972}. The second claim follows from the first
claim and Lemma 4 of \cite{arXiv:1403.3972}. 
\end{proof}

Since $1.4142^n > 7.9995^{n/6} > 3^{n/6}$, the last result in the 
above proposition improves the third result in Proposition 1 of
\cite{arXiv:1403.3972}. In Theorem 16 of \cite{arXiv:1801.10436}, we improve 
the second results of Theorem 1 and Proposition 1 of \cite{arXiv:1403.3972}:
$$
\Omega\big(2^{n/3}/(n\sqrt{n})\big) = \Omega(1.9999^{n/3}) = \Omega(1.2599^n) 
$$
The first results of Theorem 1 and Proposition 1 of \cite{arXiv:1403.3972}
will be improved by theorem \ref{subsc} and theorem \ref{sub} below,
respectively.

\subsection{Binary PFAs}

We adapt some techniques of \cite{arXiv:1801.10436} to construct 
binary PFAs with large subset synchronization.

\begin{theorem} \label{csubsc}
There exists a transitive PFA $M$ with $n$ states and only $2$ symbols, 
which has a subset which synchronizes carefully in 
$$
\Omega \big((3-\epsilon)^{n/3}\big) = \Omega(1.4422^n)
$$
steps.
\end{theorem}

\begin{proof}
We first construct a PFA $M$ with $n = 3k$ states and $3$ symbols.
The state set is $\{A_i,X_i,B_i \mid i \in \Z/(k\Z)\}$. We leave out the states
$C_i$ and the starting symbol $s$. Instead, we add a finishing symbol $f$.
\begin{align*}
(A_i c,X_i c,B_i c) &= (A_{i+1},X_{i+1},B_{i+1}) \\
(A_i r,X_i r,B_i r) &= \begin{cases}
(\bot,\bot,A_i), & \text{if $i = 1, 2, \ldots, h$} \\
(X_i,B_i,\bot), & \text{if $i = h+1$} \\
(A_i,X_i,B_i), & \text{if $i = h+2, h+3, \ldots, k$}
\end{cases} \\
(A_i f,X_i f,B_i f) &= \begin{cases}
(A_i,X_i,B_k), & \text{if $i = 1$} \\
(\bot,\bot,\bot), & \text{if $i = 2$} \\
(A_i,X_i,B_i), & \text{if $i = 3, 4, \ldots, k$}
\end{cases}
\end{align*}
The subset we start with is $\{A_{h+1},A_{h+2},\ldots,A_k\}$, so $k - h$ 
groups are represented with a state. Using techniques of 
\cite{arXiv:1801.10436}, we can prove that it takes 
$\Omega \big((3-\epsilon)^{n/3}\big)$ $r$-steps to reduce to 
$k-h-1$ such groups.

Using Theorem 9 of \cite{arXiv:1801.10436}, we reduce $\{f, c, rc\}$ 
to only $2$ symbols. Condition (1) of this theorem is not fulfilled
for $s = f$, but this is not a problem because we may choose the
subset we start with. By removing up to two states from $Q'$, we can
obtain binary automata for every number of states.
\end{proof}

\subsection{Binary DFAs}

We combine theorem \ref{csubsc} with techniques of \cite{arXiv:1403.3972} 
to construct binary DFAs with exponential subset synchronization.

\begin{theorem} \label{sub}
There exists a DFA $M$ with $n$ states and only $2$ symbols, 
which has a subset which synchronizes carefully in 
$$
\Omega \big((3-\epsilon)^{n/3}\big) = \Omega(1.4422^n)
$$
steps. Furthermore, the DFA is transitive up to $2$ sink states.
\end{theorem}

\begin{proof}
We first construct a PFA $M$ with $n = 3k + 2$ states and $3$ symbols.
The state set is $\{A_i,X_i,B_i \mid i \in \Z/(k\Z)\} \cup \{D,\bar{D}\}$. 
With two exceptions, we take the definitions in the proof of theorem 
\ref{csubsc} of $c$, $r$, $f$ for all $i \in \Z/(k\Z)$, and extend them 
by imposing that $D$ and $\bar{D}$ are sink states. 

The first exception is that we replace $\bot$ by $\bar{D}$ in all definitions.
The second exception is that we replace $B_1 f = B_k$ by $B_1 f = D$. Just as 
above, the reduction to only $2$ symbols can be done without affecting the 
estimate. The subset we start with is $\{A_{h+1},A_{h+2},\ldots,A_k,D\}$.
\end{proof}

\begin{theorem} \label{subsc}
There exists a transitive DFA $M$ with $n$ states and only $2$ symbols, 
which has a subset which synchronizes in 
$$
\Omega \big((3-\epsilon)^{n/6}\big) = \Omega(1.2009^n)
$$
\end{theorem}

\begin{proof}
We first construct a PFA $M$ with $n = 6k + 4$ states and $3$ symbols. 
The state set is 
$$
\{A_i,\bar{A}_i,X_i,\bar{X}_i,B_i,\bar{B}_i \mid i \in \Z/(k\Z)\} 
\cup \{D,\bar{D},E,\bar{E}\}
$$
By way of swap congruence, we define how symbols act on states with bars above 
them. So we only need to describe how symbols act on states without bars above 
them. Symbols $c$ and $r$ are now defined by $E c = E$, $Er = E$,
and their definitions in the proof of theorem \ref{sub}. We define
symbol $f$ as follows
\begin{align*}
Df &= E \qquad \qquad Ef = B_k \\
(A_i f,X_i f,B_i f) &= \begin{cases}
(\bar{E},\bar{E},D), & \text{if $i = 1$} \\
(\bar{E},\bar{E},\bar{E}), & \text{if $i = 2$ or $i = h+1$} \\
(A_i,X_i,B_i), & \text{if $i = 3, 4, \ldots, h, h+2, h+3, \ldots, k$}
\end{cases}
\end{align*}
The subset we start with is $\{A_{h+1},A_{h+2},\ldots,A_k,D\}$ or 
$\{A_{h+1},A_{h+2},\ldots,A_k,D,\allowbreak E\}$. 
Getting both a state and its corresponding swap state in the subset 
makes synchronization impossible. Just as in theorem \ref{csubsc},
$f$ can only be applied near the end, because otherwise we get both 
$E$ and $\bar{E}$ in our subset. Just as before, the reduction to 
only $2$ symbols can be done without affecting the estimate. 
\end{proof}

\section{D3-directing NFAs}
\label{d3nfa2pfa}

We say that an NFA $\A = (Q,\Sigma,\cdot)$ is D3-directing, if there exists a 
word $w \in \Sigma^{*}$, such that $\bigcap_{q \in Q} qw \ne \varnothing$. 
We denote the length of the shortest such $w$ by $d_3(\A)$.

For symbols $a, b$ of an NFA $\A = (Q,\Sigma,\cdot)$, we say that $a = b$ if 
$qa = qb$ for all $q$, and $a \le b$ if $qa \subseteq qb$ for all 
$q \in Q$. Furthermore, we say that $a < b$ if $a \le b$ and $a \ne b$.

We say that a symbol $a$ of an NFA $\A = (Q,\Sigma,\cdot)$ is a 
\emph{PFA-symbol}, if $|qa| \le 1$ for all $q \in Q$, and define
$$
\splt(\A) := \big(Q,\{\mbox{$x$ is a PFA-symbol} \mid 
\mbox{there is an $y \in \Sigma$ such that $x \le y$}\},\cdot\big)
$$
Notice that $\splt(\A)$ is defined as an NFA which is actually a 
PFA. For an NFA $\A = (Q,\Sigma,\cdot)$, we define 
$$
\bsc(\A) := \big(Q,\{x \in \Sigma \mid \mbox{$x \nleq \id_Q$ and 
$x \nless y$ for all $y \in \Sigma$}\},\cdot\big)
$$

\begin{lemma}
$d_3(\B) = d_3(\bsc(\B))$.
\end{lemma}

\begin{proof}
Since the symbols of $\bsc(\B)$ are a subset of those of $\B$, 
$d_3(\B) \le d_3(\bsc(\B))$ follows.

Suppose that $w$ is a D3-directing word of $\B$.
If all letters of $w$ are symbols of $\bsc(\B)$, then we are done,
so assume that $w$ has a symbol $x$ which is not a symbol of $\bsc(\B)$.
\begin{compactitem}

\item If $x \le \id_Q$, then we can remove all occurences of $x$ in $w$.

\item If $x < y$ for some symbol $y$ of $\B$, then we can replace
all occurences of $x$ by $y$ in $w$.

\end{compactitem}
We cannot repeat the above forever, so we have a procedure to 
change $w$ into a D3-directing word with only symbols of $\bsc(\B)$,
which is not longer than $w$. Hence $d_3(\B) \ge d_3(\bsc(\B))$.
\end{proof}

\begin{theorem} \label{D3}
$d_3(\A) = d_3(\splt(\A))$.
\end{theorem}

\begin{proof}
Adopt a total ordering on the state set $Q$ of $\A$.
Let $w = w_1 w_2 \cdots w_l$ be a word of $\A$, and suppose
that $w$ is D3-directing. Say that $r \in qw$ for all $q \in Q$. 

Take any $q \in Q$ and define $p_{q0} = q$. Assume by induction
that $p_{qi} \in q w_1 w_2 \cdots w_i$, and that 
$r \in p_{qi} w_{i+1} w_{i+2} \cdots w_l$. Then there exists
a $p_{q(i+1)} \in p_{qi}w_{i+1}$, such that 
$r \in p_{q(i+1)} w_{i+2} w_{i+3} \cdots w_l$.
Choosing $p_{q(i+1)}$ to be as small as possible with 
respect to the ordering of $Q$ yields an inductive definition of 
$p_{qi}$ for all $q \in Q$ and all $i \in \{1,2,\ldots,l\}$.
Furthermore, $p_{ql} = r$ for all $q \in Q$.

Since we used the ordering of $Q$, we have
$$
p_{qi} = p_{q'i} \Longrightarrow p_{q(i+1)} = p_{q'(i+1)}
$$
Consequently, the following is a proper definition of $w'_i$ for 
all $i \in \{1,2,\ldots,l\}$, where $P_i = \{p_{qi} \mid q \in Q\}$.
\begin{align*}
p_{qi} w'_i &= \{p_{q(i+1)}\} \quad (q \in Q) \\
q w'_i &= \varnothing \quad (q \in Q \setminus P_i)
\end{align*}
Let $w' = w'_1 w'_2 \cdots w'_l$. Then $qw' = \{p_{ql}\} = \{r\}$
for all $q \in Q$. As $w'_i$ is a PFA-symbol and $w'_i \le w_i$ 
for all $i$, we see that $w'$ is a D3-directing word of $\splt(\A)$.

So $d_3(\A) \ge d_3(\splt(\A))$. Conversely, if 
$w' = w'_1 w'_2 \cdots w'_l$  is a D3-directing word, say that 
$r \in qw'$ for all $q \in Q$, and $w'_i \le w_i$ for all $i$, then 
$r \in qw$ for all $q \in Q$, where $w = w_1 w_2 \cdots w_l$. So 
$d_3(\A) \le d_3(\splt(\A))$.
\end{proof}

Let 
\begin{align*}
d_3(n) &:= \max\{d_3(\A) \mid \mbox{$\A$ is an NFA}\} \\
cd_3(n) &:= \max\{d_3(\A) \mid \mbox{$\A$ is a complete NFA}\} \\
p(n) &:= \max\{d_3(\A) \mid \mbox{$\A$ is a PFA}\} \\
d(n) &:= \max\{d_3(\A) \mid \mbox{$\A$ is a DFA}\}
\end{align*}
For a PFA, D3-directing is the same as careful synchronization, which
is just synchronization in case of a DFA. So $p(n)$ and $d(n)$ are
just the maximum (careful) synchronization lengths for PFAs and
DFAs respectively of $n$ states.

\begin{corollary}
$d_3(n) = p(n)$ and $cd_3(n) = d(n)$.
\end{corollary}

\begin{proof}
Clearly, $d_3(n) \ge p(n)$ and $cd_3(n) \ge d(n)$. From the above
lemma and theorem, it follows that it suffices to show that 
$\bsc(\splt(\A))$ is actually a DFA for every complete NFA $\A$. 
This is straightforward.
\end{proof}

In \cite{arXiv:1703.07995}, in which $cd_3(n) = d(n)$ is proved as well, there 
is a description of $\splt$ which is more restrictive in choosing symbols 
in the following sense: a PFA symbol $x \le y$ is only selected if for all 
states $q$, $q x = \varnothing$ if and only if $q y = \varnothing$. 
Due to this, their $\splt$ maps complete NFAs directly to DFAs. 
But their description of $\splt$ is kind of algorithmic and takes 
pages. This is not preferable compared to an algebraic description 
of only one line. 

Lemma 3 of \cite{d3nfa} implies that $d_3(n) = p(n)$. Propositions 2 and 10 in 
\cite{d3nfa} however indicate that $cd_3(n) = d(n)$ was missed by the authors.
The proof of Lemma 3 in \cite{d3nfa} follows that of theorem \ref{D3} above 
more or less, except for adopting a total ordering on the state set,
which is $Q$ in theorem \ref{D3} above and $S$ in Lemma 3 in \cite{d3nfa}.
Using a total ordering on $S$ is needed to make that the 
partial functions $\rho_1, \rho_2, \ldots, \rho_r$ in the proof
of Lemma 3 of \cite{d3nfa} are properly defined, i.e.\@ do not depend 
on the state $s \in S$.

Let $\A$, $\C$ be NFAs. We say that $\C \le \A$ if for every symbol 
$x$ of $\C$, there is a symbol $y$ of $\A$ such that $x \le y$. 
We say that $\C < \A$ if $\C \le \A$ and $\A \nleq \C$. 
Notice that 
$$
\C \le \A \Longrightarrow d_3(\A) \ge d_3(\C)
$$
Furthermore,
$$
\splt(\A) \le \A \qquad \mbox{and} \qquad \bsc(\A) \le \A 
$$
and $\splt(\A) < \A$, if and only if $\A$ is not actually a PFA. 

Let $\B$ be an NFA, with alphabet $\Sigma$. Let $P$ be a partition
of $\Sigma$ into $p$ subsets. Then we can merge the symbols of each
of the subsets of $P$, to obtain an NFA $\C$ with at most $p$ symbols. 
We call $\C$ a {\em partitional symbolic merge of $\B$}. Notice that
$\B \le \C$. 

\begin{proposition}
Let $\A, \B$ be NFAs, such that $\B \le \A$. Then there exists a 
partitional symbolic merge $\C$ of $\B$ for which $\B \le \C \le \A$, 
such that the number of symbols of $\C$ does not exceed that of $\A$.

Furthermore, $d_3(\C) = d_3(\B) = d_3(\A)$ if $d_3(\B) \le d_3(\A)$.
\end{proposition}

\begin{proof}
For each symbol $x$ of $\B$, we choose a symbol $y$ of $\A$ such that $x \le y$.
We make $\C$ from $\B$ by merging symbols $x$ for which we chose the same symbol
$y$ of $\A$. The last claim follows from $d_3(\B) \ge d_3(\C) \ge d_3(\A)$.
\end{proof}

For $n = 2,3,4,5,6,7$, we scanned all PFAs with maximum careful synchronization 
lengths $p(n)$ which were minimal as such with respect to $\le$, and tried every 
partitional symbolic merge of it (up to straightforward pruning by way of merges 
which decrease the length of the D3-directing word). This yielded several NFAs with 
$n-1$ symbols which were D3-directing in $p(n)$ steps, but no such NFAs with fewer 
symbols. From the above proposition, we infer that NFAs with $n \le 7$ states and 
fewer than $n-1$ symbols are not D3-directing in the maximum number of 
$p(n)$ steps.

Below on the left, there is a PFA with $4$ states and $4$ symbols, of which
$a'$ is the identity symbol, which takes $p(4) = 10$ steps to synchronize
carefully. It remains D3-directing in exactly $10$ steps if the symbols 
$a$ and $a'$ are merged. This merge yields an NFA with $4$ states and $3$ 
symbols.
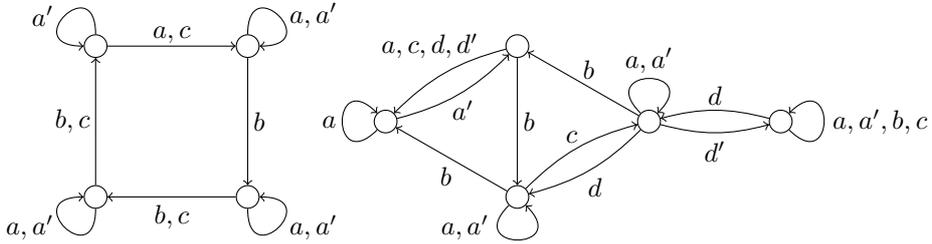
\begin{center}
	\begin{tikzpicture}
	\useasboundingbox (-1,-0.5) rectangle (3,2.5);
	\node[circle,draw,inner sep=0pt,minimum width=3mm] (1) at (0,2) {};
	\node[circle,draw,inner sep=0pt,minimum width=3mm] (2) at (2,2) {};
	\node[circle,draw,inner sep=0pt,minimum width=3mm] (3) at (2,0) {};
	\node[circle,draw,inner sep=0pt,minimum width=3mm] (4) at (0,0) {};
	\draw[->] (1) -- node[above,inner sep=1pt] {$a,c$} (2);
	\draw[->] (2) -- node[right,inner sep=2pt] {$b$} (3);
	\draw[->] (3) -- node[below,inner sep=3pt] {$b,c$} (4);
	\draw[->] (4) -- node[left,inner sep=2pt] {$b,c$} (1);
	\draw[->] (1) edge[out=90,in=-180,looseness=10] node[left,inner sep=4pt] {$a'$} (1);
	\draw[->] (2) edge[out=90,in=0,looseness=10] node[right,inner sep=4pt] {$a,a'$} (2);
	\draw[->] (3) edge[out=-90,in=0,looseness=10] node[right,inner sep=4pt] {$a,a'$} (3);
	\draw[->] (4) edge[out=-90,in=-180,looseness=10] node[left,inner sep=4pt] {$a,a'$} (4);
	\end{tikzpicture}
	\begin{tikzpicture}
	\useasboundingbox (-0.7,-0.5) rectangle (7,2.5);
	\node[circle,draw,inner sep=0pt,minimum width=3mm] (0) at (0,1) {};
	\node[circle,draw,inner sep=0pt,minimum width=3mm] (1) at (1.732,2) {};
	\node[circle,draw,inner sep=0pt,minimum width=3mm] (2) at (1.732,0) {};
	\node[circle,draw,inner sep=0pt,minimum width=3mm] (3) at (3.464,1) {};
	\node[circle,draw,inner sep=0pt,minimum width=3mm] (4) at (5.196,1) {};
	\draw[->] (0) edge[out=15,in=-135] node[pos=0.55,below,inner sep=3pt] {$a'$} (1);
	\draw[->] (1) edge[out=-165,in=45] node[pos=0.65,above,inner sep=9pt] {$a,c,d,d'$} (0);
	\draw[->] (1) -- node[right,inner sep=2pt] {$b$} (2);
	\draw[->] (2) -- node[pos=0.55,below,inner sep=4pt] {$b$} (0);
	\draw[->] (3) -- node[pos=0.45,above,inner sep=3pt] {$b$} (1);
	\draw[->] (2) edge[out=45,in=-165] node[pos=0.45,above,inner sep=3pt] {$c$} (3);
	\draw[->] (3) edge[out=-135,in=15] node[pos=0.45,below,inner sep=4pt] {$d$} (2);
	\draw[->] (4) edge[out=165,in=15] node[above,inner sep=2pt] {$d$} (3);
	\draw[->] (3) edge[out=-15,in=-165] node[below,inner sep=3pt] {$d'$} (4);
	\draw[->] (0) edge[out=225,in=135,looseness=10] node[left,inner sep=2pt] {$a$} (0);
	\draw[->] (2) edge[out=225,in=315,looseness=10] node[pos=0.2,left,inner sep=3pt] {$a,a'$} (2);
	\draw[->] (3) edge[out=135,in=45,looseness=10] node[above,inner sep=2pt] {$a,a'$} (3);
	\draw[->] (4) edge[out=315,in=45,looseness=10] node[right,inner sep=3pt] {$a,a',b,c$} (4);
	\end{tikzpicture}
\end{center}
Above on the right, there is a PFA with $5$ states and $6$ symbols, which takes 
$p(5) = 21$ steps to synchronize carefully. It remains D3-directing in exactly 
$21$ steps if the symbols $a$ and $a'$ are merged, and the symbols $d$ and $d'$ 
are merged. These merges yield an NFA with $5$ states and $4$ symbols.
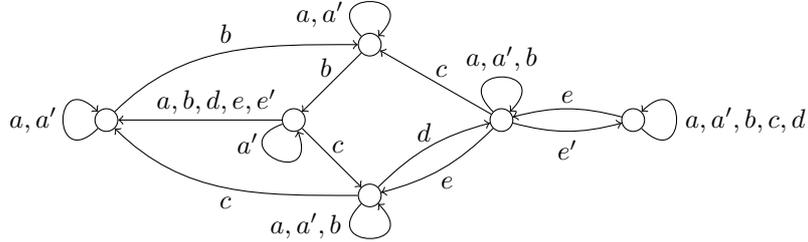
\begin{center}
	\begin{tikzpicture}
	\useasboundingbox (-1.2,-0.5) rectangle (9.1,2.5);
	\node[circle,draw,inner sep=0pt,minimum width=3mm] (0) at (0,1) {};
	\node[circle,draw,inner sep=0pt,minimum width=3mm] (1) at (2.464,1) {};
	\node[circle,draw,inner sep=0pt,minimum width=3mm] (2) at (3.464,2) {};
	\node[circle,draw,inner sep=0pt,minimum width=3mm] (3) at (3.464,0) {};
	\node[circle,draw,inner sep=0pt,minimum width=3mm] (4) at (5.196,1) {};
	\node[circle,draw,inner sep=0pt,minimum width=3mm] (5) at (6.928,1) {};
	\draw[->] (0) edge[out=45,in=-180] node[above]{$b$} (2);
	\draw[->] (1) -- node[pos=0.40,above,inner sep=2pt]{$a,b,d,e,e'$} (0);
	\draw[->] (1) -- node[pos=0.60, above,inner sep=4pt]{$c$} (3);
	\draw[->] (2) -- node[pos=0.60,above,inner sep=4pt]{$b$} (1);
	\draw[->] (3) edge[out=-180,in=-45] node[below]{$c$} (0);
	%\draw[->] (3) edge[out=15,in=-135] node[pos=0.55,below,inner sep=4pt]{$d$} (4);
	\draw[->] (3) edge[out=45,in=-165] node[pos=0.45,above,inner sep=3pt] {$d$} (4);
	\draw[->] (4) -- node[pos=0.45,above,inner sep=3pt]{$c$}(2);
	%\draw[->] (4) edge[out=-165,in=45] node[pos=0.55,above,inner sep=4pt]{$e$} (3);
	\draw[->] (4) edge[out=-135,in=15] node[pos=0.45,below,inner sep=4pt] {$e$} (3);
	\draw[->] (4) edge[out=-15,in=-165] node[below,inner sep=3pt] {$e'$} (5);
	\draw[->] (5) edge[out=165,in=15] node[above,inner sep=2pt] {$e$}(4);
	\draw[->] (0) edge[out=225,in=135,looseness=10] node[left,inner sep=2pt] {$a,a'$} (0);
	\draw[->] (2) edge[out=135,in=45,looseness=10] node[pos=0.2,left,inner sep=2pt] {$a,a'$} (2);
	\draw[->] (3) edge[out=225,in=315,looseness=10] node[pos=0.2,left,inner sep=3pt] {$a,a',b$} (3);
	\draw[->] (4) edge[out=135,in=45,looseness=10] node[above,inner sep=2pt] {$a,a',b$} (4);
	\draw[->] (5) edge[out=315,in=45,looseness=10] node[right,inner sep=3pt] {$a,a',b,c,d$} (5);
	\draw[->] (1) edge[out=200,in=290,looseness=10] node[pos=0.25,left,inner sep=1pt] {$a'$} (1);
	\end{tikzpicture}
\end{center}
Above, there is a PFA with $6$ states and $7$ symbols, of which
$a'$ is the identity symbol, which takes $p(6) = 37$ steps to synchronize
carefully. It remains D3-directing in exactly $37$ steps if the symbols 
$a$ and $a'$ are merged, and the symbols $e$ and $e'$ are merged.
These merges yield an NFA with $6$ states and $5$ symbols.
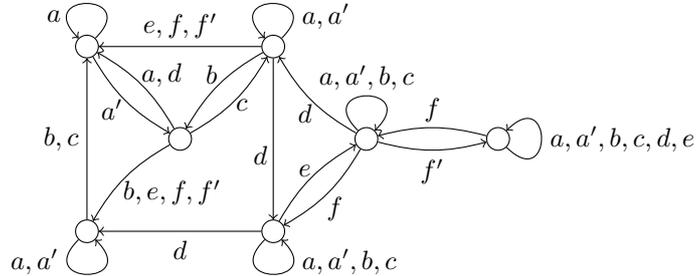
\begin{center}
	\begin{tikzpicture}[x=1.225cm,y=1.225cm]
	\useasboundingbox (-0.8,-1.4) rectangle (6.4,1.4);
	\node[circle,draw,inner sep=0pt,minimum width=3mm] (0) at (0,1) {};
	\node[circle,draw,inner sep=0pt,minimum width=3mm] (1) at (1,0) {};
	\node[circle,draw,inner sep=0pt,minimum width=3mm] (2) at (0,-1) {};
	\node[circle,draw,inner sep=0pt,minimum width=3mm] (3) at (2,1) {};
	\node[circle,draw,inner sep=0pt,minimum width=3mm] (4) at (2,-1) {};
	\node[circle,draw,inner sep=0pt,minimum width=3mm] (5) at (3,0) {};
	\node[circle,draw,inner sep=0pt,minimum width=3mm] (6) at (4.414,0) {};
	\draw[->] (3) -- node[above,inner sep=3pt]{$e,f,f'$} (0);
	\draw[->] (1) edge[in=-30,out=120] node[pos=0.2,above,inner sep=8pt]{$a,d$} (0);
	\draw[->] (0) edge[out=-60,in=150] node[pos=0.3,below,inner sep=5pt]{$a'$} (1);
	\draw[->] (2) -- node[left,inner sep=3pt]{$b,c$} (0);
	\draw[->] (1) edge[out=-150,in=60] node[pos=0.7,right,inner sep=4pt]{$b,e,f,f'$} (2);
	\draw[->] (1) edge[out=30,in=-120] node[pos=0.6,below,inner sep=2pt]{$c$} (3);
	\draw[->] (3) edge[in=60,out=-150] node[pos=0.6,above,inner sep=3pt]{$b$} (1);
	\draw[->] (3) -- node[pos=0.6,left,inner sep=2pt]{$d$} (4);
	\draw[->] (4) -- node[below]{$d$} (2);
	%\draw[->] (4) edge[out=15,in=-135] node[pos=0.55,below,inner sep=4pt]{$f$} (5);
	\draw[->] (4) edge[out=60,in=-150] node[pos=0.4,above,inner sep=3pt] {$e$} (5);
	\draw[->] (5) edge[out=150,in=-60] node[pos=0.6,below,inner sep=4pt]{$d$} (3);
	%\draw[->] (5) edge[out=-165,in=45] node[pos=0.55,above,inner sep=4pt]{$g$} (4);
	\draw[->] (5) edge[out=-120,in=30] node[pos=0.4,below,inner sep=4pt] {$f$} (4);
	\draw[->] (5) edge[out=-15,in=-165] node[below,inner sep=3pt] {$f'$} (6);
	\draw[->] (6) edge[out=165,in=15] node[above,inner sep=2pt] {$f$} (5);
	\draw[->] (0) edge[out=45,in=135,looseness=10] node[pos=0.8,left,inner sep=2pt] {$a$} (0);
	\draw[->] (3) edge[out=45,in=135,looseness=10] node[pos=0.2,right,inner sep=3pt] {$a,a'$} (3);
	\draw[->] (2) edge[out=225,in=315,looseness=10] node[pos=0.2,left,inner sep=3pt] {$a,a'$} (2);
	\draw[->] (4) edge[out=225,in=315,looseness=10] node[pos=0.8,right,inner sep=3pt] {$a,a',b,c$} (4);
	\draw[->] (5) edge[out=135,in=45,looseness=10] node[above,inner sep=2pt] {$a,a',b,c$} (5);
	\draw[->] (6) edge[out=315,in=45,looseness=10] node[right,inner sep=3pt] {$a,a',b,c,d,e$} (6);
	\end{tikzpicture}
\end{center}
Above, there is a PFA with $7$ states and $8$ symbols, which takes 
$p(7) = 63$ steps to synchronize carefully. It remains D3-directing in exactly 
$63$ steps if the symbols $a$ and $a'$ are merged, and the symbols $f$ and $f'$ 
are merged. These merges yield an NFA with $7$ states and $6$ symbols.

For $n = 5$ and $n = 7$ states, the minimum number
of $n-1$ symbols is not possible if $\splt$ contains the identity symbol. 
For $n = 2,3,4,6$ states, $\splt$ may contain the identity symbol without affecting 
the minimum number of $n-1$ symbols.
Up to the identity symbol $a'$ in case of $4$ or $6$ states, the displayed PFAs are 
the same as those in \cite{arXiv:1801.10436}.

Since the symbols of an NFA (which is actually a DFA) can be merged in other ways than 
by way of a partitional symbolic merge (the partition may be replaced by any cover), 
the above techniques do not help if you want to count the number of (basic) NFAs with 
$n$ states which are D3-directing in the maximum number of steps. 
In \cite{arXiv:1703.07995}, the basic complete NFAs with $n$ states which are D3-directing 
in the maximum number of steps are counted for all $n \le 7$, and for $n \ge 8$ under 
the assumption that the Cerny automaton with $n$ states is the only automaton which
synchronizes in at least $(n-1)^2$ steps. 

I have made the same counts and can confirm the results of \cite{arXiv:1703.07995}. 
In the case of $2$ states, there are some subtleties with symmetry to take into account: 
there are $33$ basic complete NFAs $\A$ for which $\bsc(\splt(\A))$ is one of the $6$ 
basic DFAs which synchronize, $27$ basic complete NFAs $\A$ for which $\bsc(\splt(\A))$ 
is one of the $4$ symmetrically different basic DFAs which synchronize, and $20$ 
symmetrically different basic complete NFAs $\A$ for which $\bsc(\splt(\A))$ synchronizes. 
Synchronization always takes $1$ step.

\section{An estimate on both sides of the prime number constructions of Martyugin}
\label{martyuginprime}

My co-author Henk Don of \cite{arXiv:1801.10436} reinvented the prime number 
construction in \cite{primepfa} of ternary PFAs with superpolynomial 
synchronization, which is as follows.
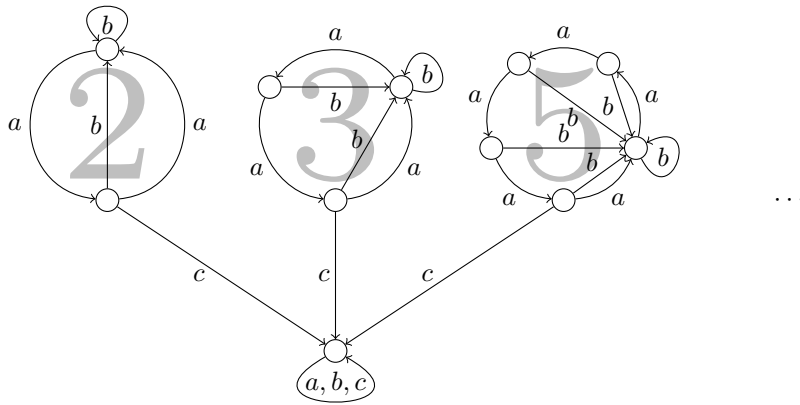
\begin{center}
\begin{tikzpicture}
\node[circle,draw,inner sep=0pt,minimum width=3mm] (S) at (3,-3) {};
\draw[->] (S) edge[out=210,in=330,looseness=15,pos=0.5,inner sep=2.5pt] node[above] {$a,b,c$} (S);
\begin{scope}[shift={(0,0)}]
\node[lightgray,scale=9] at (0,0) {$\scriptstyle 2$};
\node[circle,draw,inner sep=0pt,minimum width=3mm] (0) at (-90:1) {};
\node[circle,draw,inner sep=0pt,minimum width=3mm] (1) at (90:1) {};
\foreach \n [remember=\n as \nlast (initially 1)] \n in {0,1} {
\draw[->,bend right=84,looseness=1.5] (\nlast) edge node[shift={(180*\n-180:2mm)}] {$a$} (\n);
}
\draw[->] (1) edge[out=45,in=135,looseness=10] node[below] {$b$} (1);
\draw[->] (0) edge node[pos=0.50,left,inner sep=2pt] {$b$} (1);
\draw[->] (0) edge node[pos=0.50,left,inner sep=6pt] {$c$} (S);
\end{scope}
\begin{scope}[shift={(3,0)}]
\node[lightgray,scale=9] at (0,0) {$\scriptstyle 3$};
\node[circle,draw,inner sep=0pt,minimum width=3mm] (0) at (-90:1) {};
\node[circle,draw,inner sep=0pt,minimum width=3mm] (1) at (30:1) {};
\node[circle,draw,inner sep=0pt,minimum width=3mm] (2) at (150:1) {};
\foreach \n [remember=\n as \nlast (initially 2)] \n in {0,1,2} {
\draw[->,bend right=54] (\nlast) edge node[shift={(120*\n-150:2mm)}] {$a$} (\n);
}
\draw[->] (1) edge[out=-15,in=75,looseness=10] node[left,pos=0.4] {$b$} (1);
\draw[->] (2) edge node[pos=0.50,below,inner sep=2pt] {$b$} (1);
\draw[->] (0) edge node[pos=0.55,left,inner sep=3pt] {$b$} (1);
\draw[->] (0) edge node[pos=0.50,left,inner sep=2pt] {$c$} (S);
\end{scope}
\begin{scope}[shift={(6,0)}]
\node[lightgray,scale=9] at (0,0) {$\scriptstyle 5$};
\node[circle,draw,inner sep=0pt,minimum width=3mm] (0) at (-90:1) {};
\node[circle,draw,inner sep=0pt,minimum width=3mm] (1) at (-18:1) {};
\node[circle,draw,inner sep=0pt,minimum width=3mm] (2) at (54:1) {};
\node[circle,draw,inner sep=0pt,minimum width=3mm] (3) at (126:1) {};
\node[circle,draw,inner sep=0pt,minimum width=3mm] (4) at (198:1) {};
\foreach \n [remember=\n as \nlast (initially 4)] \n in {0,1,2, 3, 4} {
\draw[->,bend right=30] (\nlast) edge node[shift={(72*\n-126:2mm)}] {$a$} (\n);
}
\draw[->] (1) edge[out=-63,in=27,looseness=10] node[left,pos=0.55] {$b$} (1);
\draw[->] (2) edge node[pos=0.50,left,inner sep=3pt] {$b$} (1);
\draw[->] (3) edge node[pos=0.65,left,inner sep=5pt] {$b$} (1);
\draw[->] (4) edge node[pos=0.50,above,inner sep=2pt] {$b$} (1);
\draw[->] (0) edge node[pos=0.35,above,inner sep=3pt] {$b$} (1);
\draw[->] (0) edge node[pos=0.50,left,inner sep=6pt] {$c$} (S);
\end{scope}
\node at (9,-1) {$\cdots$};
\end{tikzpicture}
\end{center}
The shortest synchronizing word is $b a^{p-1} c$, where $p$ is the product of
the prime numbers which are represented by a group of states.

In \cite{primepfa}, there is a prime number construction of binary PFAs with 
superpolynomial synchronization as well, which is slightly modified below.
The difference is that $a$ maps the node right from the lowest of a prime 
group to the sink state here, instead of the node left from it. This improves
the synchronization length with $2$. 
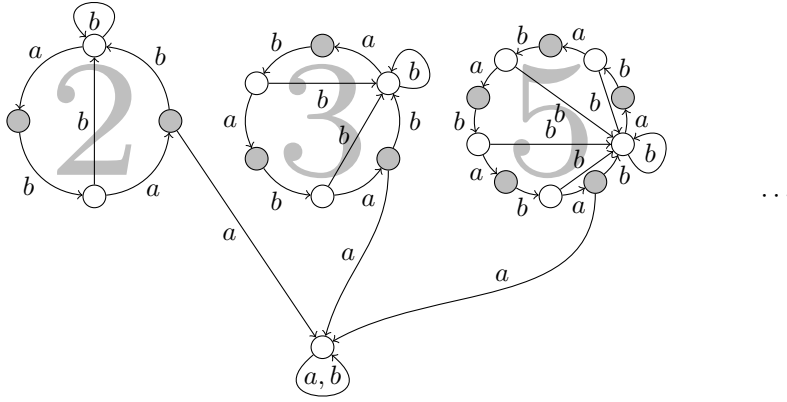
\begin{center}
\begin{tikzpicture}
\node[circle,draw,inner sep=0pt,minimum width=3mm] (S) at (3,-3) {};
\draw[->] (S) edge[out=220,in=320,looseness=12,pos=0.5,inner sep=2.5pt] node[above] {$a,b$} (S);
\begin{scope}[shift={(0,0)}]
\node[lightgray,scale=9] at (0,0) {$\scriptstyle 2$};
\node[circle,draw,inner sep=0pt,minimum width=3mm] (0) at (-90:1) {};
\node[circle,draw,fill=lightgray,inner sep=0pt,minimum width=3mm] (0p) at (0:1) {};
\node[circle,draw,inner sep=0pt,minimum width=3mm] (1) at (90:1) {};
\node[circle,draw,fill=lightgray,inner sep=0pt,minimum width=3mm] (1p) at (180:1) {};
\foreach \n [remember=\n as \nlast (initially 1)] \n in {0,1} {
\draw[->,bend right=39] (\nlast p) edge node[shift={(180*\n-135:2mm)}] {$b$} (\n);
\draw[->,bend right=39] (\n) edge node[shift={(180*\n-75:2mm)}] {$a$} (\n p);
}
\draw[->] (1) edge[out=45,in=135,looseness=10] node[below] {$b$} (1);
\draw[->] (0) edge node[pos=0.50,left,inner sep=2pt] {$b$} (1);
\draw[->] (0p) edge node[left,inner sep=3.5pt] {$a$} (S);
\end{scope}
\begin{scope}[shift={(3,0)}]
\node[lightgray,scale=9] at (0,0) {$\scriptstyle 3$};
\node[circle,draw,inner sep=0pt,minimum width=3mm] (0) at (-90:1) {};
\node[circle,fill=lightgray,draw,inner sep=0pt,minimum width=3mm] (0p) at (-30:1) {};
\node[circle,draw,inner sep=0pt,minimum width=3mm] (1) at (30:1) {};
\node[circle,fill=lightgray,draw,inner sep=0pt,minimum width=3mm] (1p) at (90:1) {};
\node[circle,draw,inner sep=0pt,minimum width=3mm] (2) at (150:1) {};
\node[circle,fill=lightgray,draw,inner sep=0pt,minimum width=3mm] (2p) at (210:1) {};
\foreach \n [remember=\n as \nlast (initially 2)] \n in {0,1,2} {
\draw[->,bend right=24] (\nlast p) edge node[shift={(120*\n-120:2mm)}] {$b$} (\n);
\draw[->,bend right=24] (\n) edge node[shift={(120*\n-60:2mm)}] {$a$} (\n p);
}
\draw[->] (1) edge[out=-15,in=75,looseness=10] node[left,pos=0.4] {$b$} (1);
\draw[->] (2) edge node[pos=0.50,below,inner sep=2pt] {$b$} (1);
\draw[->] (0) edge node[pos=0.55,left,inner sep=3pt] {$b$} (1);
\draw[->] (0p) edge[out=-90,in=75] node[pos=0.50,left,inner sep=3pt] {$a$} (S);
\end{scope}
\begin{scope}[shift={(6,0)}]
\node[lightgray,scale=9] at (0,0) {$\scriptstyle 5$};
\node[circle,draw,inner sep=0pt,minimum width=3mm] (0) at (-90:1) {};
\node[circle,fill=lightgray,draw,inner sep=0pt,minimum width=3mm] (0p) at (-54:1) {};
\node[circle,draw,inner sep=0pt,minimum width=3mm] (1) at (-18:1) {};
\node[circle,fill=lightgray,draw,inner sep=0pt,minimum width=3mm] (1p) at (18:1) {};
\node[circle,draw,inner sep=0pt,minimum width=3mm] (2) at (54:1) {};
\node[circle,fill=lightgray,draw,inner sep=0pt,minimum width=3mm] (2p) at (90:1) {};
\node[circle,draw,inner sep=0pt,minimum width=3mm] (3) at (126:1) {};
\node[circle,fill=lightgray,draw,inner sep=0pt,minimum width=3mm] (3p) at (162:1) {};
\node[circle,draw,inner sep=0pt,minimum width=3mm] (4) at (198:1) {};
\node[circle,fill=lightgray,draw,inner sep=0pt,minimum width=3mm] (4p) at (234:1) {};
\foreach \n [remember=\n as \nlast (initially 4)] \n in {0,1,2, 3, 4} {
\draw[->,bend right=12] (\nlast p) edge node[shift={(72*\n-108:2mm)}] {$b$} (\n);
\draw[->,bend right=12] (\n) edge node[shift={(72*\n-72:2mm)}] {$a$} (\n p);
}
\draw[->] (1) edge[out=-63,in=27,looseness=10] node[left,pos=0.55] {$b$} (1);
\draw[->] (2) edge node[pos=0.50,left,inner sep=3pt] {$b$} (1);
\draw[->] (3) edge node[pos=0.65,left,inner sep=5pt] {$b$} (1);
\draw[->] (4) edge node[pos=0.50,above,inner sep=2pt] {$b$} (1);
\draw[->] (0) edge node[pos=0.35,above,inner sep=3pt] {$b$} (1);
\draw[->] (0p) edge[out=-90,in=30] node[pos=0.50,above,inner sep=3pt] {$a$} (S);
\end{scope}
\node at (9,-1) {$\cdots$};
\end{tikzpicture}
\end{center}
The shortest synchronizing word is $b^2 (ab)^{p-1} a^2 = b (ba)^p a$, 
where $p$ is the product of the prime numbers which are represented by a group 
of states.

In \cite{d3nfa}, the author Martyugin found distinct lower bounds for the synchronization 
lenghs of the above ternary and binary construction, but a closer look reveals that the 
bounds are actually the same, and strictly weaker than $2^{\Omega(\sqrt{n})}$. We improve
this lower bound and also obtain an upper bound which is subexponential, by showing
that the number of required steps is 
$$
2 ^ {\Theta{\textstyle(}\sqrt{n \cdot \log n}{\textstyle)}} 
= n ^ {\Theta{\textstyle(}\sqrt{n / \log n}{\textstyle)}}
$$
We can follow the proof of Martyugin to some extend. Suppose that
the first $r$ primes are represented by a group in one of the above constructions.
Then the $i$-th prime is $\Theta(i \log i)$, which is $\Omicron(i \log r)$ and 
$\Omega\big(i^2 (\log r) / r\big)$ if $i \le r$, because $i \mapsto (\log i) / i$ 
is decreasing beyond Euler's number. Hence
$$
n = \sum_{i=1}^r \Omicron(i \log r) 
  = \sum_{i=1}^r \Omicron(i \log n) = \Omicron(r^2 \log n)
$$
So $r = \Omega(\sqrt{n/\log n})$ and $\log r = \Omega(\log n)$. Hence 
$$
n = \sum_{i=1}^r \Omega\big(i^2 (\log r) / r\big) 
  = \sum_{i=1}^r \Omega\big(i^2 (\log n) / r\big) = \Omega(r^2 \log n)
$$
So $r = \Omicron(\sqrt{n/\log n})$. Hence $r = \Theta(\sqrt{n/\log n})$. 
Furthermore,
$$
r! = \sqrt{\textstyle\prod_{i=1}^r \big(i\cdot(r+1-i)\big)} 
   \ge \sqrt{\textstyle\prod_{i=1}^r r} = r^{r/2}
$$
so
$$
p \ge r! = r^{\Omega(r)} = 2^{\Omega(r \log r)} 
  = 2^{\Omega{\textstyle(}\sqrt{n/\log n} \cdot \log n{\textstyle)}}
$$
and 
$$
p \le n^r = n^{\Omicron(r)} = 2^{\Omicron(r \log n)} 
  = 2^{\Omicron{\textstyle(}\sqrt{n/\log n} \cdot \log n{\textstyle)}}
$$
which yields the estimate.

\end{document}